\documentclass[a4paper,11pt]{article}
\usepackage[dvips,colorlinks=true,linkcolor=blue]{hyperref}
\usepackage{amssymb,amsfonts}
\usepackage{amsmath,amsthm}
\usepackage{graphicx}
\usepackage[rflt]{floatflt}

\numberwithin{equation}{section}

\newcommand{\D}{\displaystyle}

\newcommand{\R}{{\mathbb R}}
\newcommand{\T}{{\mathbb T}}
\newcommand{\Z}{{\mathbb Z}}
\newcommand{\N}{{\mathbb N}}

\newcommand{\car}{{\bold 1}}

\theoremstyle{plain}
\newtheorem{Th}{Theorem}[section]
\newtheorem{Le}{Lemma}[section]
\newtheorem{Pro}{Proposition}[section]

\theoremstyle{definition}

\title{Spectral extrema and Lifshitz tails for non monotonous alloy
  type models}
\author{Fr{\'e}d{\'e}ric Klopp\footnote{%
LAGA, U.M.R. 7539 C.N.R.S, Institut Galil{\'e}e,
Universit{\'e} de Paris-Nord, 99 Avenue J.-B.  Cl{\'e}ment, F-93430
Villetaneuse, France\ et \  Institut Universitaire de France.  
Email: {\tt klopp@math.univ-paris13.fr}} 
\ and Shu Nakamura\footnote{%
Graduate School of Mathematical Sciences, University of Tokyo, 
3-8-1 Komaba, Meguro-ku, Tokyo, Japan 153-8914. 
Email: {\tt shu@ms.u-tokyo.ac.jp}}} 
\begin{document}
\maketitle
\begin{abstract} In the present note, we determine the ground state
  energy and study the existence of Lifshitz tails near this energy
  for some non monotonous alloy type models.  Here, non monotonous
  means that the single site potential coming into the alloy random
  potential changes sign. In particular, the random operator is not a
  monotonous function of the random variables.
  \vskip.5cm\noindent
  \textsc{R{\'e}sum{\'e}.} Cet article est consacr{\'e} {\`a} la d{\'e}termination de
  l'{\'e}nergie de l'{\'e}tat fondamental et {\`a} l'{\'e}tude de possibles
  asymptotiques de Lifshitz au voisinage de cette {\'e}nergie pour
  certains mod{\`e}les d'Anderson continus non monotones. Ici, non
  monotone signifie que le potentiel de simple site entrant dans la
  composition du potentiel al{\'e}atoire change de signe. En particulier,
  l'op{\'e}rateur al{\'e}atoire n'est pas une fonction monotone des variables
  al{\'e}atoires.
\end{abstract}
\setcounter{section}{-1}
\section{Introduction and results}
\label{intro}
In this paper, we consider the continuous alloy type (or Anderson)
random Schr{\"o}dinger operator:
\begin{equation}
  \label{eq:1}
  H_\omega=-\Delta+V_\omega\text{ where
  }V_\omega(x)=\sum_{\gamma\in\Z^d}\omega_\gamma V(x-\gamma)
\end{equation}
on $\R^d$, $d\geq 1$, where $V$ is the site potential, and
$(\omega_\gamma)_{\gamma\in\Z^d}$ are the random coupling constants.
Throughout this paper, we assume
\begin{description}
\item[(H1)]
  \begin{enumerate}
  \item $V:\ \R^d\to\R$ is $L^p$ (where $p=2$ if $d\leq 3$ and $p>d/2$
    if $d>3$), non identically vanishing and supported in
    $(-1/2,1/2)^d$;
  \item $(\omega_\gamma)_\gamma$ are independent identically
    distributed (i.i.d.) random variables distributed in $[a,b]$
    ($a<b$) with essential infimum $a$ and essential supremum $b$.
  \end{enumerate}
\end{description}
Let $\Sigma$ be the almost sure spectrum of $H_\omega$ and
$E_-=\inf\Sigma$. When $V$ has a fixed sign, it is well known that the
$E_-=\inf(\sigma(-\Delta+V_{\overline{b}}))$ if $V\leq0$ and
$E_-=\inf(\sigma(-\Delta+V_{\overline{a}}))$ if $V\geq0$. Here,
$\overline{x}$ is the constant vector
$\overline{x}=(x)_{\gamma\in\Z^d}$.\\
Moreover, in this case, it is well known that the integrated density
of states of the Hamiltonian (see e.g.~\eqref{eq:4}) admits a Lifshitz
tail near $E_-$, i.e., that the integrated density of states at energy
$E$ decays exponentially fast as $E$ goes to $E_-$ from above. We
refer
to~\cite{MR99c:82038,MR87c:82066,MR1935594,MR94h:47068,Ki-Ma:83b,Ki:89,MR2000m:82029}
for precise statements.\vskip.2cm\noindent
In the present paper, we address the case when $V$ changes sign, i.e.,
there may exist $x_+\not=x_-$ such that
\begin{equation}
  \label{eq:17}  V(x_-)\cdot V(x_+)<0.
\end{equation}
The basic difficulty this property introduces is that the variations
of the potential $V_\omega$ as a function of $\omega$ are not
monotonous. In the monotonous case, to get the minimum, one can simply
minimize with respect to each of the random variables individually. In
the non monotonous case, this uncoupling between the different random
variables may fail. Our results concern reflection symmetric
potentials since, as we will see, for these potentials we also have an
analogous decoupling between the different random variables. Thus, we
make the following symmetry assumption on $V$:
\begin{description}
\item[(H2)] $V$ is reflection symmetric i.e. for any
  $\sigma=(\sigma_1,\dots,\sigma_d)\in\{0,1\}^d$ and any
  $x=(x_1,\dots,x_d)\in\R^d$,
  \begin{equation*}
    V(x_1,\dots,x_d)=V((-1)^{\sigma_1}x_1,\dots,(-1)^{\sigma_d}x_d).
  \end{equation*}
\end{description}
We now consider the operator $H_\lambda^N=-\Delta+\lambda V$ with
Neumann boundary conditions on the cube $[-1/2,1/2]^d$. Its spectrum
is discrete, and we let $E_-(\lambda)$ be its ground state energy. It
is a simple eigenvalue and $\lambda\mapsto E_-(\lambda)$ is a real
analytic concave function defined on $\R$.  We first observe:
\begin{Pro}
  \label{pro:1}
  Under the above assumptions (H1) and (H2),
  \begin{equation*}
    E_-=\inf(E_-(a),E_-(b)).
  \end{equation*}
\end{Pro}
\noindent For $a$ and $b$ sufficiently small, this result was proven
in~\cite{MR2202533} without the assumption (H2) but with an additional
assumption on the sign of $\D\int_{\R^d}V(x)dx$. The method used by
Najar relies on a small coupling constant expansion for the infimum of
$\Sigma$.  These ideas were first used in~\cite{MR2335580} to treat
other non monotonous perturbations, in this case magnetic ones, of the
Laplace operator. In~\cite{BaLoSto:07}, the authors study the minimum
of the almost sure spectrum for a random displacement model i.e. the
random potential is defined as
$V_\omega(x)=\sum_{\gamma\in\Z^d}V(x-\gamma-\xi_\gamma)$ where
$(\xi_\gamma)_\gamma$ are i.i.d. random variables supported in a
sufficiently small compact.\\
We now turn to the results on Lifshitz tails. We denote by $N(E)$ the
integrated density of states of $H_\omega$, i.e., it is defined by the
limit
\begin{equation}
  \label{eq:4}
  N(E)=\lim_{L\to+\infty}\frac{\#\{\text{eigenvalues of
    }H_{\omega,L}^N\ \leq E\}}{(2L+1)^d}
\end{equation}
where $H_{\omega,L}^N$ is the operator $H_\omega$ restricted to the
cube $[-L-1/2,L+1/2]^d$ with Neumann boundary conditions. This limit
exists for a.s.\ $\omega$ and is independent of $\omega$; it has been
the object of a lot of studies and we refer
to~\cite{MR94h:47068,MR2051995,MR2378428} for extensive reviews.

We first give an upper bound on the integrated density of states. In
the applications of the Lifshitz tails asymptotics, in particular, to
localization, this side of the bound is the most important and also
the difficult one to obtain.
\begin{Th}
  \label{thr:2}
  Suppose assumptions (H1) and (H2) are satisfied. Assume moreover
  that
  \begin{equation}
    \label{eq:11}
    E_-(a)\not=E_-(b).
  \end{equation}
  Then
  \begin{equation}
    \label{eq:2}
    \limsup_{E\to E_-^+}\frac{\log|\log N(E)|}{\log(E-E_-)}\leq
    -\frac d2-\alpha_+
  \end{equation}
  where we have set $c=a$ if $E_-(a)<E_-(b)$ and $c=b$ if
  $E_-(a)>E_-(b)$, and
  \begin{equation*}
    \alpha_+=-\frac12\liminf_{\varepsilon\to0}\frac{\log|\log\mathbb{P}(\{
      |c-\omega_0|\leq\varepsilon\})|}{\log\varepsilon}\geq0.
  \end{equation*}
\end{Th}
\noindent As will be clear from the proofs, we could also consider the
model $H_\omega=H_0+V_\omega$ where $V_\omega$ is as above and
$H_0=-\Delta+W$ where $W$ is a $\Z^d$-periodic potential that
satisfies the symmetry assumption (H2).
\vskip.1cm\noindent We now study a lower bound for the integrated
density of states that we will prove in a more general case than the
upper bound, i.e., we don't need to assume (H2). The assumption is as
follows:
\begin{description}
\item[(HP)] there exists $\omega^P\in[a,b]^{\Z^d}$ that is periodic (,
  i.e., for some $L_0\in\N$, for all $\gamma\in\Z^d$ and
  $\beta\in\Z^d$, $\omega^P_{\gamma+L_0\beta}=\omega^P_\gamma$) such
  that $\inf\Sigma=\inf\sigma(H_{\omega^P})$.
\end{description}
Under this assumption, we have
\begin{Th}
  \label{thr:3}
  Let $H_\omega$ be defined as above, and assume (H1) and (HP) hold.
  Then
  \begin{equation}
    \label{eq:9}
    -\frac{d}2-\alpha_-\leq
    \liminf_{E\to E_-^+}\frac{\log|\log N(E)|}{\log(E-E_-)},
  \end{equation}
  where
  \begin{equation*}
    \alpha_-=-\liminf_{\varepsilon\to0}\frac{\log|\log\mathbb{P}(\{
      \forall\gamma\in\Z^d/L_0\Z^d;\ |\omega^P_\gamma-\omega_\gamma|
      \leq\varepsilon\})|}{\log\varepsilon}.
  \end{equation*}
\end{Th}
\noindent Assume now that (H1) and (H2) hold, and hence, (HP) holds
with $\omega^P=\overline{a}$ or $\omega^P=\overline{b}$. Indeed, as we
will see in the proof of Proposition~\ref{pro:1} in
Section~\ref{sec:determ-bott-spectr}, under assumption (H2), $E_-(a)$
is also the bottom of the spectrum of the periodic operator
$H_{\overline{a}}$.  If we assume that $\alpha_+=\alpha_-=0$ and
$E_-(a)\not= E_-(b)$, we obtain the following corollary:
\begin{Th}
  \label{thr:5}
  Assume that (H1) and (H2) and~\eqref{eq:11} hold and that
  $\alpha_-=\alpha_+=0$ then
  \begin{equation*}
    \lim_{E\to E_-^+}\frac{\log|\log N(E)|}{\log(E-E_-)}=-\frac{d}2.
  \end{equation*}
\end{Th}
\noindent Combining Theorem~\ref{thr:2} with the Wegner estimates
obtained in~\cite{Kl:95b,Hi-Kl:01a} and the multiscale analysis as
developed in~\cite{MR2002m:82035}, we learn
\begin{Th}
  \label{thr:1}
  Assume (H1), (H2) and~\eqref{eq:11} hold. Assume, morevoer, that the
  common distribution of the random variables admits an absolutely
  continuous density. Then, the bottom edge of the spectrum of
  $H_\omega$ exhibits complete localization in the sense
  of~\cite{MR2002m:82035}.
\end{Th}
\noindent Lifshitz tail have already been proved for various non
monotonous random models, mainly models with a random magnetic fields
(see e.g.~\cite{MR2335580,MR2015430,MR1806979,MR1800861}). In the
models we consider, we will now see that Lifshitz tails do not always
appear.
\vskip.1cm\noindent In a companion paper (see~\cite{FN:08}), we study
the case when $E_-(a)=E_-(b)$. This requires techniques different from
the ones used in the present paper and gets particularly interesting
when the random variables are Bernoulli distributed. However, it is
quite easy to see that, when $E_-(a)=E_-(b)$, Lifshitz tails may fail;
the density of states can even exhibit a van Hove singularity.
\begin{Th}
  \label{thr:7}
  There exists potentials $V$ and random variables
  $(\omega_\gamma)_\gamma$ sa\-tisfying (H1) and (H2) such that
  \begin{itemize}
  \item $E_-(a)=E_-(b)=0$
  \item there exists $C>0$ such that, for $E\geq 0$, one has
    \begin{equation*}
      \frac1C E^{d/2}\leq N(E)\leq C E^{d/2}.
    \end{equation*}
  \end{itemize}
\end{Th}
\noindent This paper is constructed as follows. In
Section~\ref{sec:determ-bott-spectr}, we determine the bottom of the
almost sure spectrum, and prove Proposition~\ref{pro:1}. In
Section~\ref{sec:lifshitz-tails}, we prove our main theorems,
Theorem~\ref{thr:3} and Theorem~\ref{thr:1}. Finally, in
section~\ref{sec:an-example-where}, we prove Theorem~\ref{thr:7}.
%
%
\vskip.2cm\noindent\textbf{Acknowledgment.} FK would like to thank the
University of Tokyo where part of this work was done. A part of this
research was done when SN was invited to Univ.\ Paris 13 in 2007, and
he would like to thank it. SN is partially supported by JSPS Research
Grant, Kiban (B) 17340033.


\section{Determining the bottom of the spectrum}
\label{sec:determ-bott-spectr}
We denote by $t\mapsto E_-(t)$ the ground state energy of the operator
$H_{t,0}^N$, i.e., $-\Delta+tV$ on $[-1/2,1/2]^d$ with Neumann
boundary conditions.
We first note that $E_-(t)$ is a concave function of $t$ as, by the
variational principle, it is the infimum of a family of affine
functions of $t$.
Hence, for $t\in[a,b]$, we have $E_-(t)\geq\min(E_-(a),E_-(b))$. Then,
partitioning $\R^d$ into the cubes $\gamma+[-1/2,1/2]^d$ for
$\gamma\in\Z^d$, and, restricting $H_\omega$ to each of these cubes
with Neumann boundary conditions, we obtain
\begin{equation*}
  H_\omega\geq\bigoplus_{\gamma\in\Z^d}H_{\omega_\gamma,0}^N
\end{equation*}
So, we learn
\begin{equation*}
  H_\omega\geq\min(E_-(a),E_-(b)).
\end{equation*}
We let $L\geq1$, and consider $H_{\omega,L}^P$, the operator
$H_\omega$ restricted to the cube $[-L-1/2,L+1/2]^d$ with periodic
boundary conditions. Clearly, this operator depends only on finitely
many random variables. We prove
\begin{Le}
  \label{le:1}
  \begin{equation*}
    \Sigma=\overline{\bigcup_{L\geq1}\bigcup_{\omega\text{
          admissible}}\sigma(H_{\omega,L}^P)}, 
  \end{equation*}
  where $\omega$ is called admissible if all the components of
  $\omega$ are in the support of the distribution of the random
  variables defining the alloy type random operator.
\end{Le}

This lemma is a variant of a standard characterization of the almost
sure spectrum of an alloy type model.  To prove
Proposition~\ref{pro:1}, i.e., that $E_-=\min(E_-(a),E_-(b))$, it is
hence enough to prove that, for any $L$ sufficiently large,
\begin{equation}
  \label{eq:5}
  \inf_{\omega\in[a,b]^{C_L^d}}\inf\sigma(H_{\omega,L}^P)\leq\min(E_-(a),E_-(b))
\end{equation}
where $C_L^d=\Z^d\cap[-L-1/2,L+1/2]^d$.  To prove~\eqref{eq:5}, we
will use the assumption (H.2). For the sake of definiteness, let us
assume $E_-(a)\leq E_-(b)$.

The ground state of $H_{a,0}^N$, say $\psi$, is simple and can be
chosen uniquely as a normalized positive function. The reflection
symmetry of the potential $V$ guarantees that $\psi$ is reflection
symmetric. For $\gamma\in\Z^d$ such that
$|\gamma|=|\gamma_1|+\cdots+|\gamma_d|=1$, we can continue $\psi$ to
the $\gamma+[-1/2,1/2]^d$ by reflection symmetry with respect to the
common boundary of $[-1/2,1/2]^d$ and $\gamma+[-1/2,1/2]^d$. As $\psi$
is reflection symmetric, we continue this process of reflection with
respect to the boundaries of the new cubes to obtain a continuation of
$\psi$ that is $\Z^d$-periodic, positive and reflection symmetric with
respect to any plane that is common boundary to two cubes of the form
$\gamma+[-1/2,1/2]^d$. Moreover, $\psi$ satisfies, for any $L\geq0$,
$H_{\overline{a},L}^P\psi=H_{a,0}^P\psi=H_{a,0}^N\psi=E_-(a)\psi$.
This proves that $E_-(a)\geq \inf\sigma(H_{a,L}^P)$.
Hence,~\eqref{eq:5} holds.  This completes the proof of
Proposition~\ref{pro:1}.\qed
\begin{proof}[Proof of Lemma~\ref{le:1}]
  \label{sec:proof-lemma}
  Recall a well known characterization of the almost sure spectrum of
  an alloy type model in terms of periodic approximations (see
  e.g.~\cite{MR94h:47068,Ki:89}). Therefore, for $L\geq1$, define the
  $L\Z^d$-periodic operator
  \begin{equation}
    \label{eq:10}
    H_{\omega,L}=-\Delta+V_{\omega,L},\quad V_{\omega,L}(\cdot)=\sum_{\beta\in
      L\Z^d} \sum_{\gamma\in\Z^d/(L\Z^d)} \omega_\gamma
    V(\cdot-\beta-\gamma).
  \end{equation}
  Then, one has
  \begin{equation*}
    \Sigma=\overline{\bigcup_{L\geq1}\bigcup_{(\omega_\gamma)_{\gamma\in\Z^d/(L\Z^d)}
        \text{ admissible}}\sigma(H_{\omega,L})}
  \end{equation*}
  where $(\omega_\gamma)_{\gamma\in\Z^d/(L\Z^d)}$ is admissible if all
  its components belong to the support of the random variables
  defining
  the alloy type model. \\
  Floquet theory (see e.g.~\cite{MR58:12429c}) guarantees that
  $\sigma(H^P_{\omega,L})\subset \sigma(H_{\omega,L})$. So, in order
  to prove Lemma~\ref{le:1}, it is sufficient to prove that
  \begin{equation}
    \label{eq:6}
    \sigma(H_{\omega,L})\subset \overline{\bigcup_{n\geq1}\sigma(H^P_{\omega^L,nL})}
  \end{equation}
  for some well chosen admissible $\omega^L$.

  Consider $\omega^L$ defined by
  $\omega^L_{\gamma+L\beta}=\omega_\gamma$ for $\gamma\in\Z^d/(L\Z^d)$
  and $\beta\in\Z^d$. Clearly, $V_{\omega^L}=V_{\omega,L}$; hence, if
  $E_n(\theta)$ are the Floquet eigenvalues of $H_{\omega,L}$, the
  spectrum of the operator $H^P_{\omega^L,nL}$ is the set
  $\{E_n(2\pi\gamma/n)$; $\gamma\in\Z^d/(nL\Z^d)\}$ (see
  e.g.\cite{MR2003k:82049}). The inclusion \eqref{eq:6} follows from
  the continuity of the Floquet eigenvalues as function of the Floquet
  parameter (see e.g Lemma 7.1 in~ \cite{MR1979772}).
\end{proof}
%

\section{Lifshitz tails}
\label{sec:lifshitz-tails}
To fix ideas let us assume $E_-(a)< E_-(b)$. The two bounds in
Theorem~\ref{thr:2} and Theorem~\ref{thr:3} will be proved separately.
\subsection{The upper bound}
\label{sec:upper-bound}
The upper bound on the integrated density of states,
Theorem~\ref{thr:2}, will be an immediate consequence of the following
result.
\begin{Th}
  \label{thr:6}
  Suppose assumptions (H1) and (H2) are satisfied, and, that
  $E_-(a)<E_-(b)$. Then, there exists $c>0$ such that, for $E\geq
  E_-(a)$, one has
  \begin{equation}
    \label{eq:3}
    N(E)\leq N_m(C(E-E_-(a)))
  \end{equation}
  where $N_m$ is the integrated density of states of the random
  operator
  \begin{equation}
    \label{eq:12}
    H^m_{\omega}=H_{\overline{a}}-E_-(a)+
    \sum_{\gamma\in\Z^d}(\omega_\gamma-a)\car_{[-1/2,1/2]^d}(x-\gamma)
  \end{equation}
  and $H_{\overline{a}}$ is defined above.
\end{Th}
\noindent The upper bound is then deduced from the same bound for the
integrated density of states of $H^m_{\omega}$ which is standard, see
e.g.~\cite{MR94h:47068,Ki:89,MR1935594} and references therein.
\begin{proof}
  We first note it is well known that, at $E$, a continuity point of
  $N(E)$, the sequence
  \begin{equation}
    \label{eq:13}
    N_L^N(E)=\mathbb{E}\left(\frac{\#\{\text{eigenvalues of
        }H_{\omega,L}^N\ \leq E\}}{(2L+1)^d}\right)
  \end{equation}
  is decreasing and converges to $N(E)$ (see
  e.g.~\cite{MR94h:47068,Ki:89}). So to prove Theorem~\ref{thr:6}, it
  suffices to prove that, there exists $C>0$ such that, for $E$ real
  and $L$ large
  \begin{equation*}
    N_L^N(E)\leq N_L^{m}(C(E-E_-(a)))
  \end{equation*}
  where $N_L^{m}(E)$ is defined by~\eqref{eq:13} where
  $H_{\omega,L}^N$ is replaced by $H_{\omega,L}^m$ , i.e., the
  restriction of $H_{\omega}^m$
  to $[-L-1/2,L+1/2]^d$ with Neumann boundary conditions.\\
  By the Rayleigh-Ritz principle (\cite{MR58:12429c}, Section~XIII.1),
  this follows from the quadratic form inequality
  \begin{equation*}
    H_{\omega,L}^m\leq C (H_{\omega,L}^N-E_-(a)).
  \end{equation*}
  Under our assumptions on $V$, the form domain of both of these
  operators is $H^1([-L-1/2,L+1/2]^d)$. Now, as for $\psi\in
  H^1([-L-1/2,L+1/2]^d)$ and $\gamma\in\Z^d\cap[-L-1/2,L+1/2]^d$,
  $\psi\car_{\gamma+[-1/2,1/2]^d}\in H^1(\gamma+[-1/2,1/2]^d)$, this
  inequality in turns follows from the inequalities
  \begin{gather*}
    \forall\gamma\in\Z^d\cap+[-L-1/2,L+1/2]^d,\ \forall\psi\in
    H^1(\gamma+[-1/2,1/2]^d), \\
    \langle H_{\omega,L}^m\psi,\psi\rangle_{\gamma+[-1/2,1/2]^d} \leq
    C\langle
    (H_{\omega,L}^N-E_-(a))\psi,\psi\rangle_{\gamma+[-1/2,1/2]^d}.
  \end{gather*}
%
  Note that, here, the choice of the boundary condition is crucial:
  the form domain of the Neumann operator is the whole $H^1$-space;
  moreover, the Neumann quadratic form does not involve boundary
  terms.
  Taking into account the structure of our random potentials, we see
  that this will follow from the operator inequality
  \begin{equation}
    \label{eq:14}
    (H^N_{a,0}-E_-(a))+(t-a) \leq C (H_{t,0}^N-E_-(a)), \quad t\in [a,b], 
  \end{equation}
  where $H_{t,0}^N=-\Delta+tV$ on $[-1/2,1/2]^d$ with Neumann boundary
  conditions as before.
  \begin{Le}
    \label{le:3}
    Let $H_0$ be self-adjoint on $\mathcal{H}$ a separable Hilbert
    space such that $0=\inf\sigma(H_0)$. Let $V_1$ be a closed
    symmetric operator relatively bounded with respect to $H_0$ with
    bound $0$.  Set $H_1=H_0+V_1$ and $E_1=\inf\sigma(H_1)$. Assume
    $E_1>0$. Then, there exists $C>0$ such that, for $t\in[0,1]$, one
    has
    \begin{equation*}
      C(H_0+tV_1)\geq H_0 + t. 
    \end{equation*}
  \end{Le}
  \noindent Lemma~\ref{le:3} applies to our case and, as a result, we
  obtain~\eqref{eq:14}. This completes the proof of
  Theorem~\ref{thr:6}.
\end{proof}
\begin{proof}[Proof of Lemma~\ref{le:3}]
  Regular perturbation theory ensures that, for some $\beta>1$, $\inf
  \sigma(H_0+\beta V_1)=\delta>0$. Then, for $t\in [0,1]$, we have
  \begin{align*}
    H_0+t V_1 &= (1-t/\beta)H_0+(t/\beta)(H_0+\beta V_1) \\
    &\geq (1-1/\beta)H_0 +(t/\beta)\delta\geq \frac1C(H_0+t)
  \end{align*}
  where $C^{-1}=\min (1-1/\beta,\delta/\beta)$.
\end{proof}
%
%
\subsection{The lower bound}
\label{sec:lower-bound}
We will use the techniques set up in~\cite{MR2003c:82037,MR1979772}.
We first recall two results from these papers which are valid in the
generality of the present work. Consider a random Schr{\"o}dinger operator
of the form~\eqref{eq:1} where
\begin{description}
\item[(HL1)] $V$ is a not identically vanishing, real valued,
  compactly supported function that is in $L^p$ (where $p=2$ if $d\leq
  3$ and $p>d/2$ if $d>3$);
\item[(HL2)] the random variables $(\omega_\gamma)_{\gamma\in\Gamma}$
  are independent, identically distributed, non trivial and bounded.
\end{description}
Clearly these two assumptions are consequences of assumption (H1).\\
The existence of $N(E)$, the integrated density of states defined
by~\eqref{eq:4} is known (\cite{MR94h:47068,MR2051995}).
\begin{Th}[\cite{MR2003c:82037}]
  \label{thr:4}
  Assume (HL1) and (HL2) hold. Pick $\eta_0>0$ and $I\subset\R$, a
  compact interval. Then, there exists $\nu_0>0$ and $\varepsilon_0>0$
  such that, for $0<\varepsilon<\varepsilon_0$, $E_-\in I$ and
  $n\geq\varepsilon^{-\nu_0}$, we have
  \begin{equation}
    \label{estpre}
    \begin{split}
      \mathbb{E}(N_{\omega,n}(E_-+&\varepsilon/2)-
      N_{\omega,n}(E_--\varepsilon/2))
      -e^{-(n\,\varepsilon^{\nu_0})^{-\eta_0}} \\\leq
      N(&E_-+\varepsilon)-N(E_--\varepsilon)\\
      &\leq \mathbb{E}(N_{\omega,n}(E_-+2\varepsilon)-
      N_{\omega,n}(E_--2\varepsilon))
      +e^{-(n\,\varepsilon^{\nu_0})^{-\eta_0}}.
    \end{split}
  \end{equation}
  where $N_{\omega,n}$ is the integrated density of states of the
  periodic operator $H_{\omega,n}$ defined in~\eqref{eq:10}.
\end{Th}
\noindent In~\cite{MR2003c:82037}, Theorem~\ref{thr:4} is not stated
in exactly the same form and under slightly stronger but unnecessary
assumptions. The modifications necessary to obtain the form given here
are simple and left to the reader.
\vskip.1cm\noindent The second result we use is the following
\begin{Le}[\cite{MR2003c:82037}]
  \label{le:4}
  Consider a periodic Schr{\"o}dinger operator of the form $H_{\omega^P}$
  where $\omega^P$ satisfies (HP). Let $E_-=\inf\sigma(H_{\omega^P})$.
  Then, for any $\alpha\in(0,1)$ and any $\varepsilon$ sufficiently
  small, there exists a function $w_\varepsilon$ with the following
  properties
  \begin{enumerate}
  \item $w_\varepsilon$ is supported in a ball of center 0 and radius
    $2\varepsilon^{-(1+2\alpha)/2}$;
  \item $1\leq \Vert w_{\varepsilon}\Vert_{L^2}^2\leq 2$;
  \item $\Vert (H_{\omega^P}-E_-) w_{\varepsilon}\Vert_{L^2}^2\leq
    C\varepsilon^{2(1+\alpha)}$ for some $C>0$ (independent of
    $\varepsilon$).
  \end{enumerate}
\end{Le}
\noindent Though not formulated as a lemma, this result is proved in
section 2 of~\cite{MR2003c:82037}.
\vskip.1cm\noindent We now prove the lower bound~\eqref{eq:9}. Pick
$\alpha\in(0,1)$ arbitrary. Let $\Gamma'$ be the lattice with respect
to which $\omega^P$ is periodic.  Pick $n\in\N^*$ such that
$n\sim\varepsilon^{-\nu}$ for $\nu>\nu_0$ ($\nu_0$ given by
Theorem~\ref{thr:4}), $\nu$ to be chosen sufficiently large. As
$w_\varepsilon$ constructed in Lemma~\ref{le:4} has compact support in
the interior of $C(0,n)$, it can be ``periodized'' to satisfy
quasi-periodic boundary conditions; for $\theta\in\R^d$, we set
\begin{equation*}
  w_{\varepsilon,\theta}(\cdot )=\sum_{\beta\in(2n+1)\Gamma'}
  e^{-i\beta\theta} w_\varepsilon(\cdot +\beta).
\end{equation*}
Then, $w_{\varepsilon,\theta}$ satisfies $w_{\varepsilon,\theta}
(x+\beta) =e^{i\beta\theta} w_{\varepsilon,\theta}(x)$ for $x\in\R^d$
and $\beta\in(2n+1)\Gamma'$, and we have
\begin{equation}
  \label{eq:24}
  \Vert w_{\varepsilon,\theta}\Vert_{L^2(C(0,n))}^2\geq1.
\end{equation}
We define
\begin{equation*}
  \Lambda_{\alpha}(\varepsilon)=\left\{\gamma\in\Gamma';\text{ for
    }1\leq j\leq d ,\ \vert\gamma_j\vert\leq
    \varepsilon^{-(1+3\alpha)/2}\right\}.
\end{equation*}
Since $V$ satisfies (HL1), if we assume
$\omega_\gamma\leq\varepsilon^{1+\alpha}$ for
$\gamma\in\Lambda_{\alpha}(\varepsilon)$, then
\begin{equation*}
  \Vert V_{\omega,n}(H_{\omega^P}-E_--1)^{-1}\Vert\leq
  C\varepsilon^{1+\alpha}.
\end{equation*}
This,~\eqref{eq:24} and point (3) of Lemma~\ref{le:4} imply that, for
some $C>0$ and $\varepsilon$ sufficiently small, we have
\begin{equation*}
  \Vert(H_{\omega,n,\theta}-E_-)w_{\varepsilon,\theta}\Vert\leq
  C\varepsilon^{1+\alpha}\Vert w_{\varepsilon,\theta}\Vert.
\end{equation*}
This proves that, for $\varepsilon$ sufficiently small, if
$\omega_\gamma\leq \varepsilon^{1+\alpha}$ for
$\gamma\in\Lambda_{\alpha}(\varepsilon)$, then for all
$\theta\in\T^*_n$, $H_{\omega,n,\theta}$ has an eigenvalue in
$[E_--\varepsilon/2,E_-+\varepsilon/2]$. Hence, we learn that
\begin{equation}
  \label{eq:15}
  \mathbb{E}(N_{\omega,n}(E_-+\varepsilon/2)-
  N_{\omega,n}(E_--\varepsilon/2))\geq \frac1C
  n^{-d}P(\varepsilon,\alpha)
\end{equation}
where $P(\varepsilon,\alpha)$ is the probability of the event
$\{\omega;\ \forall\gamma\in\Lambda_{\alpha}(\varepsilon),\
\omega_\gamma\leq\varepsilon^{1+\alpha}\}$. By assumption (HL2) and as
$n\sim\varepsilon^{-\nu}$, we have
\begin{equation}
  \label{eq:16}
  \liminf_{\substack{\varepsilon\to0\\\varepsilon>0}}
  \frac{\log\vert\log(P(\varepsilon,\alpha))\vert}{\log\varepsilon}
  \geq-\frac d2(1+3\alpha)-\alpha_-
\end{equation}
where $\alpha_-$ is defined in Theorem~\ref{thr:3}.
Plugging~\eqref{eq:16} and \eqref{eq:15} into~\eqref{estpre}, since
$\alpha>0$ is arbitrary, this completes the proof of
Theorem~\ref{thr:3}.\qed
\section{An example where Lifshitz tails fail to exist}
\label{sec:an-example-where}
We now will construct an example of the type described in
Theorem~\ref{thr:7}. \\
Let $\varphi\in\mathcal{C}^\infty((-1/2,1/2)^d)$ be positive,
reflection symmetric, constant near the boundary of $[-1/2,1/2]^d$ and
normalized in the cube. Let $V=\Delta\varphi/\varphi$; it
satisfies~\eqref{eq:17} and assumption (H2). Moreover, $\varphi$ is
the positive normalized ground state of $-\Delta+V$ on $[-1/2,1/2]^d$
with Neumann boundary conditions, the associated ground state energy
being $0$. Let $(\omega_\gamma)_{\gamma\in\Z^d}$ be Bernoulli random
variables with support $\{0,1\}$.
Pick $L\geq1$ and let $\varphi_L$ be the ground state of the operator
$H^N_{\omega,L}$ acting on $[-1/2,L-1/2]^d$ with Neumann boundary
conditions defined in Section~\ref{sec:determ-bott-spectr}. Then, this
ground state can be described as follows
\begin{itemize}
\item in $\gamma+[-1/2,1/2]^d$, $\varphi_L(\cdot)=
  L^{-d/2}\varphi(\cdot-\gamma)$ if $\omega_\gamma=1$;
\item in $\gamma+[-1/2,1/2]^d$, $\varphi_L(\cdot)=L^{-d/2} C_0$ if
  $\omega_\gamma=0$
\end{itemize}
where the constant $C_0$ is chosen to be equal to the constant value
of $\varphi$ near the boundary of $[-1/2,1/2]^d$. The thus
constructed ground state is not normalized.\\
We can now use the results of~\cite{MR89b:35127} to obtain a lower
bound on the second eigenvalue of $H^N_{\omega,L}$ that is independent
of $\omega$. Indeed, the construction above shows that, for any
$\omega$,
\begin{gather*}
  C_0\leq L^{d/2}\cdot\max_{x\in[-1/2,L-1/2]^d}\varphi_L(x)
  \leq\max_{x\in[-1/2,1/2]^d}\varphi(x),\\
  \min_{x\in[-1/2,1/2]^d}\varphi(x)\leq
  L^{d/2}\cdot\min_{x\in[-1/2,L-1/2]^d}\varphi_L(x)\leq C_0.
\end{gather*}
Then, Theorem 1.4 of~\cite{MR89b:35127} applied to $H^N_{\omega,L}$
and the Neumann Laplacian on the same cube guarantees that the second
eigenvalue of $H^N_{\omega,L}$ is larger than $cL^{-2}$ where the
constant $c$ does not depend on $\omega$, nor on $L$. The standard
upper bound for the integrated density of states by the normalized
Neumann counting function (see
e.g.~\cite{MR94h:47068,Ki:89,MR1935594}) yields, for any $L\geq 1$,
\begin{equation*}
  N(E)\leq \mathbb{E}\left(\frac{\#\{\text{eigenvalues of
      }H_{\omega,L}^N\ \leq E\}}{L^d}\right).
\end{equation*}
If we pick $L=(C^2E)^{-1/2}$ for some $C>0$ such that $cC^2\geq2$, the
second eigenvalue of $H^N_{\omega,L}$ is larger than $2E$, this for
any realization of $\omega$, we obtain
\begin{equation*}
  N(E)\leq L^{-d}= C^{d}E^{d/2}.
\end{equation*}
The standard lower bound for the integrated density of states by the
normalized Dirichlet counting function (see
e.g.~\cite{MR94h:47068,Ki:89,MR1935594}) yields
\begin{equation}
  \label{eq:8}
  \mathbb{E}\left(\frac{\#\{\text{eigenvalues of
      }H_{\omega,L}^D\ \leq E\}}{L^d}\right)\leq N(E)
\end{equation}
where $H^D_{\omega,L}$ is the Dirichlet restriction of $H_\omega$ to
$[-1/2,L-1/2]^d$.  Let $\psi_L$ be the (positive normalized) ground
state of the Dirichlet Laplacian on $[-1/2,L-1/2]^d$. And define
$\phi_L=\psi_L\cdot\varphi_L$. This smooth function clearly satisfies
Dirichlet boundary conditions and one computes, for some $C>1$,
\begin{equation}
  \label{eq:18}
  \begin{split}
    \langle H_{\omega,L}^D\phi_L,\phi_L\rangle&=
    \langle(-\Delta+V_\omega)\phi_L,\phi_L\rangle\\
    &=\langle(-\Delta\psi_L)\varphi_L,\phi_L\rangle
    -2\langle\nabla\psi_L\cdot\nabla\varphi_L,\phi_L\rangle\\
    &=2\frac{\pi^2}{L^2}\|\phi_L\|^2+\|\varphi_L\,|\nabla\psi_L|\,\|^2\leq
    \frac{C}{L^2}\|\phi_L\|^2.
  \end{split}
\end{equation}
where, in the final step, we integrated by parts and used the explicit
form of the positive normalized ground state of the Dirichlet
Laplacian.\\
Hence, taking $L=CE^{-1/2}$, as $C>1$, \eqref{eq:18} ensures that the
ground state energy of $H_{\omega,L}^D$ is less than $E$. Thus, from
\eqref{eq:8}, we learn
\begin{equation*}
  \frac1{C^d}E^{d/2}=L^{-d}\leq N(E).
\end{equation*}
Finally, combining the two estimates, we have proved that, for the
above random model, the density of states exhibits a van Hove
singularity at the bottom of the spectrum that is, there exists $C>1$
such that, for $E\geq0$, one has
\begin{equation*}
  \frac1C E^{d/2}\leq N(E)\leq C E^{d/2}.
\end{equation*}
This completes the proof of Theorem~\ref{thr:7}.\qed
%
%

\def\cprime{$'$} \def\cydot{\leavevmode\raise.4ex\hbox{.}}

\end{document}